\documentclass[letterpaper, 10 pt, conference]{ieeeconf}  

\IEEEoverridecommandlockouts                              
\usepackage{amsmath} 
\usepackage{amssymb}  
\usepackage{cite}
\usepackage{xcolor}
\usepackage{nccmath}
\usepackage{pifont}
\newcommand{\cmark}{\ding{51}}
\newcommand{\xmark}{\ding{55}}

\title{\LARGE \bf
Small-gain analysis of exponential incremental input/output-to-state stability for large-scale distributed systems
}

\author{Christian Gatke, Julian D. Schiller, Matthias A. Müller
\thanks{This work was supported by the Deutsche Forschungsgemeinschaft (DFG, German Research Foundation) - 426459964}
\thanks{Christian Gatke, Julian D. Schiller, and Matthias A. Müller are with the Leibniz University Hannover,  Institute of Automatic Control,
        30167 Hannover, Germany
        {\tt\small \{gatke, schiller, mueller\}@irt.uni-hannover.de}}%
}

\begin{document}

\maketitle
\thispagestyle{empty}
\pagestyle{empty}

\begin{abstract}
We provide a detectability analysis for nonlinear large-scale distributed systems in the sense of exponential incremental input/output-to-state stability (i-IOSS). In particular, we prove that the overall system is exponentially i-IOSS if each subsystem is i-IOSS, with interconnections treated as external inputs, and a suitable small-gain condition holds. The analysis is extended to a Lyapunov characterization, resulting in a different quantitative outcome regarding the small-gain condition, which is further analyzed within this work. Moreover, we derive linear matrix inequality conditions posed solely on the local subsystems and their interconnections, which guarantee exponential i-IOSS of the overall distributed system. The results are illustrated on a numerical example.
\end{abstract}

\section{INTRODUCTION}
Driven by advances in technologies such as 5G, the Internet of Things, and cloud computing, large-scale distributed systems are becoming increasingly prevalent in modern industry; examples include complex chemical processes, multi-robot systems, and transportation networks. Such systems usually consist of many (potentially small) subsystems, which are connected to each other either physically, by constraints, or they share the same objective as it is common for multi-agent systems (see \cite{mas} for a survey). In order to perform a proper state estimation for these kind of systems, a sophisticated analysis regarding detectability is needed. In particular, incremental input/output-to-state stability (i-IOSS) has become a standard detectability condition for nonlinear systems \cite{sontag_97}, \cite{allan_21}, \cite{schiller_23_iv}, \cite{knüfer_20}, which characterizes the bounded difference between any two trajectories of the system states with respect to their initial conditions, disturbances, and output measurements. In this sense, i-IOSS represents a detectability property, since diverging states induce observable differences in the output. In recent years, i-IOSS has received considerable research attention. For instance, \cite{allan_21} provides a converse theorem establishing the existence of an i-IOSS Lyapunov function for such systems. A time-discounted version of i-IOSS is proposed in \cite{knüfer_20}, where the influence of past disturbances and measurements decays, thereby emphasizing more recent information. Especially in the context of Moving Horizon Estimation (MHE), i-IOSS appears frequently as detectability condition \cite{arezki_23}, \cite{yang_25}, \cite{schiller_23}, \cite{knüfer_18}, \cite{knüfer_23}. In \cite{schiller_23}, e.g., a robust stability analysis for centralized MHE is provided using a Lyapunov function. Moreover, \cite{schiller_23}, \cite{arezki_lmi_23} propose linear matrix inequalities (LMIs) to verify exponential \mbox{i-IOSS} for a given system. However, these methods suffer from the curse of dimensionality, becoming intractable for large complex systems. Furthermore, the verification gets even more difficult if the network structure and/or the number of subsystems involved changes dynamically. \\
Instead of analyzing the overall system as a whole, a decentralized approach may be more reasonable, in which the subsystems are treated individually while accounting for their interconnections. To make conclusions about the overall system with these kind of approaches, small-gain concepts are often utilized \cite{jiang_94}, \cite{teel_96}, \cite{iasson_11}. The idea is to restrict the coupling gains with an appropriate condition to ensure that they remain sufficiently small such that the interconnections between the subsystems do not amplify each other. In this context, numerous contributions have appeared in the field of input-to-state stability (ISS), using small-gain theorems to establish stability properties of interconnected systems \cite{dashkovskiy_07}, \cite{dashkovskiy_11}, \cite{mironchenko_24}. In \cite{jiang_94}, a concept of input-to-output practical stability is introduced, and a small-gain condition is provided to guarantee that this property is preserved under the interconnection of two systems. As for the interconnection of two exponential ISS systems, the authors in \cite{guiver_23} take advantage of a small-gain condition to prove that this property remains valid under interconnection. A generalized version of the nonlinear small-gain theorem is proposed in \cite{dashkovskiy_07}, where more than two coupled ISS subsystems are considered. Moreover, in \cite{dashkovskiy_11}, ISS networks in a dissipative form are analyzed. As for incremental ISS, a small-gain approach is utilized in \cite{eijnden_23} to show that the feedback interconnection of a hybrid integrator and a linear time-invariant system is incrementally ISS. However, to the best of the authors' knowledge, there are no existing results in the context of \mbox{i-IOSS} regarding the small-gain theory. \\
In this work, we provide a distributed detectability analysis for large-scale nonlinear systems in the sense of exponential i-IOSS. In doing so, we impose an exponential i-IOSS assumption on the subsystems by taking the couplings into account and exploit a small-gain condition on the interconnections to prove that the overall system is exponentially \mbox{i-IOSS}. This is done for exponential i-IOSS with the classical trajectory-based formulation and in Lyapunov coordinates, where the latter indicates a less conservative small-gain condition. In addition, we extend the LMI conditions in \cite{schiller_23} to verify the exponential i-IOSS property for each subsystem in order to conclude detectability for the overall system. Moreover, with a finite set of different subsystem dynamics, the verification can be performed for networks with arbitrarily large number of subsystems. \\
The setup we consider in this work is introduced in Section \ref{sec:setup}. The analysis of exponential i-IOSS for distributed systems is provided in Section \ref{sec:IOSS} and the Lyapunov characterization is given in Section \ref{sec:lyap}. Section \ref{sec:verif} provides the verification with a comparison between the exponential \mbox{i-IOSS} property from Section \ref{sec:IOSS} and the Lyapunov characterization in Section \ref{sec:lyap}, which is illustrated on an example in Section \ref{sec:example}. Section \ref{sec:conclusion} concludes this work.

\subsection{Notation and preliminaries}
\label{sec:notation}
We denote positive integers greater than or equal to $a$ by $\mathbb{I}_{\geq a}$, while $\mathbb{I}_{[a,b]}$ stands for all positive integers in the interval $[a,b]$. The identity matrix with dimension $M$ is represented by $I_M$ and a diagonal matrix with scalar entries $\lambda_i$ is denoted by $\mathrm{diag}(\lambda_1, \cdot \cdot \cdot, \lambda_M)$. In case $\lambda_i$ are matrices, the expression indicates a block-diagonal matrix. Positive (semi)definiteness of a matrix P is denoted by $P \succ 0$ $(P \succeq 0)$, while negative (semi)definiteness is expressed by $P \prec 0$ $(P \preceq 0)$. For a vector $x \in \mathbb{R}^n$ and a matrix $P \in \mathbb{R}^{n \times n}$, $\|x\|$ and $\|P\|$ denote the Euclidean and induced 2-Norms, respectively, while $\|x\|_P^2 = x^T P x$ for $P \succ 0$. Moreover, $[x^{(i)}]_{i=1}^M$ represents a stacked vector from $1$ to $M$, i.e., $\begin{pmatrix} {x^{(1)}}^T \cdot \cdot \cdot {x^{(M)}}^T \end{pmatrix}^T$. The spectral radius of a real square matrix $G$ is given by $\rho(G)$, while $\lambda_{\min}(A)$ and $\lambda_{\max}(A)$ denote the minimum and maximum eigenvalue for a real symmetric and square matrix $A$, respectively.

\section{SETUP}
\label{sec:setup}

We consider a network of $M$ nonlinear discrete-time coupled subsystems
\vspace{-0.1em}
\begin{subequations}
	\begin{equation}
		x_{t+1}^{(i)} = f_i(x_t^{(i)}, u_t^{(i)}, w_t^{(i)}, z_t^{(i)}), \; i \in \mathbb{I}_{[1,M]},
		\label{eq:sub_dyn}
	\end{equation}
	\begin{equation}
		y_t^{(i)} = h_i(x_t^{(i)}, u_t^{(i)}, w_t^{(i)}, z_t^{(i)}), \; i \in \mathbb{I}_{[1,M]},
		\label{eq:sub_out}
	\end{equation}
	\label{eq:sub}
\end{subequations}
\vspace{-0.1em}
with the subsystem state $x_t^{(i)} \in \mathbb{R}^{n_i}$, the known control input $u_t^{(i)} \in \mathbb{R}^{m_i}$  and the disturbance $w_t^{(i)} \in \mathbb{R}^{q_i}$. Note that the process disturbance and measurement noise is, without loss of generality, summarized in one variable $w_t^{(i)}$. In~(\ref{eq:sub}), $z_t^{(i)} \in \mathbb{R}^{s_i}$ contains the system states $x_t^{(j)}$ with all ${j \in \mathcal{N}_i \subseteq \mathbb{I}_{[1,M]} \setminus \{i\}}$, where $\mathcal{N}_i$ denotes the neighbor set, containing all indices to which subsystem $i$ is coupled. The noisy output of each subsystem $i$ is indicated by ${y_t^{(i)} \in \mathbb{R}^{p_i}}$. Stacking the system states $x_t = [x_t^{(i)}]_{i=1}^M$ and doing the same for the control input $u_t = [u_t^{(i)}]_{i=1}^M$, disturbance ${w_t = [w_t^{(i)}]_{i=1}^M}$ and output $y_t = [y_t^{(i)}]_{i=1}^M$, we obtain
\vspace{-0.1em}
\begin{subequations}
	\begin{equation}
		x_{t+1} = f_{\Sigma}(x_t, u_t, w_t),
	\end{equation}
	\begin{equation}
		y_t = h_{\Sigma}(x_t, u_t, w_t)
	\end{equation}
	\label{eq:sys}
\end{subequations}
\vspace{-0.1em}
for the overall system, where $f_{\Sigma}$ and $h_{\Sigma}$ denote the stacked nonlinear dynamics $f_i$ and output functions $h_i$, respectively. \\

Next, we define the exponential i-IOSS property for nonlinear detectability.
\newtheorem{definition}{Definition}
\begin{definition}[exponential i-IOSS]
	System (\ref{eq:sys}) is said to be exponentially i-IOSS if there exist $\eta_\Sigma \in (0,1)$ and $p_\Sigma$, $q_\Sigma$, $r_\Sigma > 0$  such that
	\begin{align}
		\| x_{t} & - \tilde{x}_{t} \| \leq \eta_\Sigma^t p_\Sigma \| x_0 - \tilde{x}_0 \| \nonumber \\
		& + \sum_{k = 0}^{t-1} \eta_\Sigma^{t-k-1} \bigg( q_\Sigma \| w_{k} - \tilde{w}_{k} \| + r_\Sigma \| y_{k} - \tilde{y}_{k} \| \bigg)		
	\end{align}
	holds for all times $t \in \mathbb{I}_{\geq 0}$, all initial conditions ${x_0, \; \tilde{x}_0 \in \mathbb{R}^{n}}$, and all disturbances $w_t, \;\tilde{w}_t \in \mathbb{R}^q$, where $x_{t+1} = f_{\Sigma}(x_t, u_t, w_t)$, $\tilde{x}_{t+1} = f_{\Sigma}(\tilde{x}_t, u_t, \tilde{w}_t)$, ${y_t = h_{\Sigma}(x_t, u_t, w_t)}$, $\tilde{y}_t = h_{\Sigma}(\tilde{y}_t, u_t, \tilde{w}_t)$, $t \in \mathbb{I}_{\geq 0}$.
\label{def:IOSS}
\end{definition}

In \cite{allan_21}, a converse theorem is provided, establishing the equivalence between i-IOSS and the existence of an i-IOSS Lyapunov function. We proceed with the definition of the exponential case.

\begin{definition}[exponential i-IOSS Lyapunov function]
	System (\ref{eq:sys}) admits an exponential i-IOSS Lyapunov function $V_{\Sigma} : \mathbb{R}^n \times \mathbb{R}^n \rightarrow \mathbb{R}_{\geq 0}$ if there exist $P_{\Sigma 1}, \; P_{\Sigma 2} \succ 0$, $Q, \; R \succeq 0$ and $\lambda_\Sigma \in (0,1)$ such that
	\begin{subequations}
		\begin{equation}
			\|x_t - \tilde{x}_t\|_{P_{\Sigma 1}}^2 \leq V_\Sigma(x_t,\tilde{x}_t) \leq \|x_t - \tilde{x}_t\|_{P_{\Sigma 2}}^2,
		\end{equation}
		\begin{align}
			V_\Sigma(x_{t+1},\tilde{x}_{t+1}) & -  V_\Sigma(x_t,\tilde{x}_t) \leq -\lambda_\Sigma V_\Sigma(x_t,\tilde{x}_t) \nonumber \\
			& + \|w_{t} - \tilde{w}_{t}\|_{Q_\Sigma}^2 + \|y_{t} - \tilde{y}_{t}\|_{R_\Sigma}^2
			\label{eq:lyap_def}
		\end{align}
	\end{subequations}
	holds for all times $t \in \mathbb{I}_{\geq 0}$, all initial conditions ${x_0, \; \tilde{x}_0 \in \mathbb{R}^{n}}$, and all disturbances $w_t, \;\tilde{w}_t \in \mathbb{R}^q$, where $x_{t+1} = f_{\Sigma}(x_t, u_t, w_t)$, $\tilde{x}_{t+1} = f_{\Sigma}(\tilde{x}_t, u_t, \tilde{w}_t)$, ${y_t = h_{\Sigma}(x_t, u_t, w_t)}$, $\tilde{y}_t = h_{\Sigma}(\tilde{y}_t, u_t, \tilde{w}_t)$, $t \in \mathbb{I}_{\geq 0}$.
	\label{def:lyap}
\end{definition}


\section{NONLINEAR DETECTABILITY OF THE OVERALL SYSTEM}

\subsection{Exponential i-IOSS property}
\label{sec:IOSS}

In this section, we prove exponential i-IOSS of the overall system (\ref{eq:sys}) by considering only the subsystems with their interconnections and impose a small-gain condition to be fulfilled. To this end, we require that each individual subsystem is i-IOSS according to the following assumption.

\newtheorem{assumption}{Assumption}
\begin{assumption}
	There exist $\eta_i \in (0,1)$, $p_i, q_i, r_i, g_{ij} > 0$ such that each subsystem (\ref{eq:sub}) satisfies
	\begin{align}
		\| x_{t}^{(i)} - & \tilde{x}_{t}^{(i)} \| \leq \eta_i^t p_i \| x_0^{(i)} - \tilde{x}_0^{(i)} \| \nonumber \\
		& + \sum_{k=0}^{t-1} \eta_i^{t-k-1} \bigg( q_i \| w_k^{(i)} - \tilde{w}_k^{(i)} \| \nonumber \\
		& + r_i \| y_k^{(i)} - \tilde{y}_k^{(i)} \| + \sum_{j \in \mathcal{N}_i}^{} g_{ij} \|x_k^{(j)}-\tilde{x}_k^{(j)}\| \bigg)
		\label{eq:IOSS}
	\end{align}
		for all times $t \in \mathbb{I}_{\geq 0}$, all initial conditions $x_0^{(i)}, \; \tilde{x}_0^{(i)} \in \mathbb{R}^{n_i}$, all disturbances $w^{(i)}, \; \tilde{w}^{(i)} \in \mathbb{R}^{q_i}$, and all ${x^{(j)}, \; \tilde{x}^{(j)} \in \mathbb{R}^{n_j}}$, where $x_{t+1}^{(i)} = f_i(x_t^{(i)}, u_t^{(i)}, w_t^{(i)}, z_t^{(i)})$, $\tilde{x}_{t+1}^{(i)} = f_i(\tilde{x}_t^{(i)}, u_t^{(i)}, \tilde{w}_t^{(i)}, \tilde{z}_t^{(i)})$, $y_t^{(i)} = h_i(x_t^{(i)}, u_t^{(i)}, w_t^{(i)}, z_t^{(i)})$, $\tilde{y}_t^{(i)} = h_i(\tilde{y}_t^{(i)}, u_t^{(i)}, \tilde{w}_t^{(i)}, \tilde{z}_t^{(i)})$, $t \in \mathbb{I}_{\geq 0}$.
	\label{ass:IOSS}
\end{assumption}

In (\ref{eq:IOSS}), we treat $x^{(j)}$ and $\tilde{x}^{(j)}$ as additional independent input for subsystem $i$. Before stating Theorem \ref{theo:IOSS}, we introduce the gain matrix
\nopagebreak
\begin{equation}
	G := [(1-\eta_i)^{-1} g_{ij}]_{i,j=1, \cdot \cdot \cdot, M},
	\label{eq:G}
\end{equation}

which essentially encodes the mutual dependencies of the individual systems in the network. Note that $G$ is a square matrix whose diagonal entries are all zero.

\newtheorem{theorem}{Theorem}
\begin{theorem}[exponential i-IOSS of the overall system]
	Let Assumption \ref{ass:IOSS} hold. Then, the overall system (\ref{eq:sys}) is exponentially i-IOSS according to Definition \ref{def:IOSS} if the small-gain condition $\rho(G) < 1$ is satisfied.
	\label{theo:IOSS}
\end{theorem}

The proof of Theorem \ref{theo:IOSS} follows the approach from \cite{guiver_23}, where ISS of a system consisting of two coupled subsystems is derived utilizing a small-gain condition, and adapts it to the case of i-IOSS, i.e., the incremental version including outputs with an arbitrary number of subsystems $M$.

\begin{proof}
	Define $\Delta x := x-\tilde{x}$, $\Delta w := w-\tilde{w}$, and ${\Delta y := y-\tilde{y}}$. Now, taking the maximum of $\| \Delta w_k^{(i)} \|$, $\| \Delta y_k^{(i)} \|$, and $\| \Delta x^{(j)} \|$ over $k \in \mathbb{I}_{[0,t]}$ and utilizing the geometric series in (\ref{eq:IOSS}), we obtain
	\begin{align}
		\| \Delta & x_{t}^{(i)} \| \leq \eta_i^t p_i \| \Delta x_0^{(i)} \| + \frac{1}{1-\eta_i} \bigg( q_i \max\limits_{k \in \mathbb{I}_{[0,t]}} \| \Delta w_k^{(i)} \| \nonumber \\
		& + r_i \max\limits_{k \in \mathbb{I}_{[0,t]}} \| \Delta y_k^{(i)} \|  + \sum_{j \in \mathcal{N}_i}^{} g_{ij} \max\limits_{k \in \mathbb{I}_{[0,t]}} \| \Delta x_k^{(j)} \| \bigg).
		\label{eq:tau}
	\end{align}
	By defining $\tilde{q}_i := q_i(1-\eta_i)^{-1}$ and $\tilde{r}_i := r_i(1-\eta_i)^{-1}$, taking the supremum over $[0, \infty)$ and stacking the inequalities from index $1$ to $M$, we arrive at
	\begin{align}
		\bigg[ \sup\limits_{k \in \mathbb{I}_{\geq 0}} & \| \Delta x_k^{(i)} \| \bigg]_{i=1}^M \leq \bigg[ p_i \| \Delta x_0^{(i)} \| \bigg]_{i=1}^M \nonumber \\
		& + \bigg[ \tilde{q}_i \sup\limits_{k \in \mathbb{I}_{\geq 0}} \| \Delta w_k^{(i)} \| \bigg]_{i=1}^M + \bigg[ \tilde{r}_i \sup\limits_{k \in \mathbb{I}_{\geq 0}} \| \Delta y_k^{(i)} \| \bigg]_{i=1}^M \nonumber \\
		& + G \bigg[ \sup\limits_{k \in \mathbb{I}_{\geq 0}} \| \Delta x_k^{(i)} \| \bigg]_{i=1}^M
		\label{eq:stacked_sup}
	\end{align}
	where $G$ is defined as in (\ref{eq:G}). Note that $g_{ij} = 0$ if there is no coupling between subsystems $i$ and $j$. Now, assuming that the small-gain condition $\rho(G) < 1$ is satisfied, we can write
	\begin{align}
		\bigg[ \sup\limits_{k \in \mathbb{I}_{\geq 0}} \| \Delta x_k^{(i)} \| \bigg]_{i=1}^M \leq (I_M - G)^{-1} \bigg( \bigg[ p_i \| \Delta x_0^{(i)} \| \bigg]_{i=1}^M \nonumber \\
		+ \bigg[ \tilde{q}_i \sup\limits_{k \in \mathbb{I}_{\geq 0}} \| \Delta w_k^{(i)} \| \bigg]_{i=1}^M + \bigg[ \tilde{r}_i \sup\limits_{k \in \mathbb{I}_{\geq 0}} \| \Delta y_k^{(i)} \| \bigg]_{i=1}^M \bigg).
		\label{eq:sup_bound}
	\end{align}
	 Next, we consider the time interval $[\xi N + l, (\xi+1) N + l]$ with $N, \; l, \; \xi \in \mathbb{I}_{\geq0}$. From (\ref{eq:tau}), it follows that
	\begin{align}
		\| \Delta & x_{(\xi+1) N + l}^{(i)} \| \leq \eta_i^N p_i \| \Delta x_{\xi N + l}^{(i)} \| \nonumber \\
		& + \tilde{q}_i \max\limits_{k \in \mathbb{I}_{[\xi N + l, (\xi+1) N + l]}} \| \Delta w_k^{(i)} \| \nonumber \\
		& + \tilde{r}_i \max\limits_{k \in \mathbb{I}_{[\xi N + l, (\xi+1) N + l]}} \| \Delta y_k^{(i)} \| \nonumber \\
		& + \sum_{j \in \mathcal{N}_i}^{} \frac{g_{ij}}{1-\eta_i} \max\limits_{k \in \mathbb{I}_{[\xi N + l, (\xi+1) N + l]}} \| \Delta x_k^{(j)} \|.		
	\end{align}
	Maximizing over $l \in \mathbb{I}_{\geq 0}$, we obtain
	\begin{align}
		\sup\limits_{k \in \mathbb{I}_{\geq (\xi+1) N}} \| \Delta x_{k}^{(i)} \| & \leq \eta_i^N p_i \sup\limits_{k \in \mathbb{I}_{\geq \xi N}} \| \Delta x_{k}^{(i)} \| \nonumber \\
		& + \tilde{q}_i  \sup\limits_{k \in \mathbb{I}_{\geq 0}} \| \Delta w_k^{(i)} \| + \tilde{r}_i \sup\limits_{k \in \mathbb{I}_{\geq 0}} \| \Delta y_k^{(i)} \| \nonumber \\
		& + \sum_{j \in \mathcal{N}_i}^{} \frac{g_{ij}}{1-\eta_i}  \sup\limits_{k \in \mathbb{I}_{\geq \xi N}} \| \Delta x_k^{(j)} \|.
		\label{eq:max_l}
	\end{align}
	Recall that (\ref{eq:max_l}) holds for each subsystem $i \in \mathbb{I}_{[1,M]}$. Hence, the stacked subsystems satisfy
	\begin{align}
		\bigg[ \sup\limits_{k \in \mathbb{I}_{\geq (\xi+1) N}} \| \Delta x_{k}^{(i)} \| \bigg]_{i=1}^M & \leq S \bigg[ \sup\limits_{k \in \mathbb{I}_{\geq \xi N}} \| \Delta x_{k}^{(i)} \| \bigg]_{i=1}^M \nonumber \\
		+ \bigg[ \tilde{q}_i  \sup\limits_{k \in \mathbb{I}_{\geq 0}} \| \Delta w_k^{(i)} \| \bigg]_{i=1}^M & + \bigg[ \tilde{r}_i \sup\limits_{k \in \mathbb{I}_{\geq 0}} \| \Delta y_k^{(i)} \| \bigg]_{i=1}^M.
		\label{eq:k_comp}
	\end{align}
	with $S := I_M [\eta_i^N p_i]_{i=1}^M + G$. Note that the matrix G is the same as in (\ref{eq:stacked_sup}). Applying (\ref{eq:k_comp}) $\xi$ times, we arrive at
	\begin{align}
		\bigg[ \sup\limits_{k \in \mathbb{I}_{\geq \xi N}} \| \Delta x_{k}^{(i)} \| \bigg]_{i=1}^M & \leq S^\xi \bigg[ \sup\limits_{k \in \mathbb{I}_{\geq 0}} \| \Delta x_{k}^{(i)} \| \bigg]_{i=1}^M \nonumber \\
		& + \sum_{j=0}^{\xi-1} S^j \bigg( \bigg[ \tilde{q}_i  \sup\limits_{k \in \mathbb{I}_{\geq 0}} \| \Delta w_k^{(i)} \| \bigg]_{i=1}^M \nonumber \\
		& + \bigg[ \tilde{r}_i \sup\limits_{k \in \mathbb{I}_{\geq 0}} \| \Delta y_k^{(i)} \| \bigg]_{i=1}^M \bigg).
		\label{eq:main_bound}
	\end{align}
	Now, inserting (\ref{eq:sup_bound}) in (\ref{eq:main_bound}), we obtain
	\begin{align}
		\bigg[ & \sup\limits_{k \in \mathbb{I}_{\geq \xi N}} \| \Delta x_{k}^{(i)} \| \bigg]_{i=1}^M \leq S^\xi (I_M - G)^{-1} \bigg[ p_i \| \Delta x_0^{(i)} \| \bigg]_{i=1}^M \nonumber \\
		& + \bigg( S^\xi (I_M - G)^{-1} + \sum_{j=0}^{\xi-1} S^j \bigg) \bigg( \bigg[ \tilde{q}_i  \sup\limits_{k \in \mathbb{I}_{\geq 0}} \| \Delta w_k^{(i)} \| \bigg]_{i=1}^M \nonumber \\
		& + \bigg[ \tilde{r}_i \sup\limits_{k \in \mathbb{I}_{\geq 0}} \| \Delta y_k^{(i)} \| \bigg]_{i=1}^M \bigg).
		\label{eq:main_bound_insert}
	\end{align}
	Due to the fact that $\eta_i \in (0,1)$ and $G$ satisfies the small-gain condition (i.e., $\rho(G) < 1$), there exists $N \in \mathbb{I}_{\geq 0}$ large enough such that $\rho(S) < 1$. Hence, there exist $b > 0$ and $\sigma_0 \in (0,1)$ such that $\|S^\xi\| \leq b \sigma_0^\xi$ for all $\xi \in \mathbb{I}_{\geq 0}$. Moreover, we choose $\sigma \in (0,1)$ such that $\sigma_0^\xi \leq \sigma^{(\xi+1)N} \leq \sigma^{\xi N + l_0}$ holds with $l_0 \in \mathbb{I}_{[0,N]}$. Taking now the norm of (\ref{eq:main_bound_insert}), we arrive at 
	\begin{align}
		& \bigg\| \bigg[ \sup\limits_{k \in \mathbb{I}_{\geq \xi N}} \| \Delta x_{k}^{(i)} \| \bigg]_{i=1}^M \bigg\| \leq \bar{g} b \sigma^{\xi N + l_0} \max\limits_{i \in \mathbb{I}_{[1,M]}} p_i \| \Delta x_0 \| \nonumber \\
		& + \bar{b} \bigg( \bigg\| \bigg[ \tilde{q}_i  \sup\limits_{k \in \mathbb{I}_{\geq 0}} \| \Delta w_k^{(i)} \| \bigg]_{i=1}^M \bigg\| {+} \bigg\| \bigg[ \tilde{r}_i \sup\limits_{k \in \mathbb{I}_{\geq 0}} \| \Delta y_k^{(i)} \| \bigg]_{i=1}^M \bigg\| \bigg)
		\label{eq:main_bound_norm}
	\end{align}
	with $\bar{g} = \|(I_M-G)^{-1}\|$ and $\bar{b} = b(\bar{g} + (1-\sigma_0)^{-1})$. In the following, we make use of the general relation
	\begin{equation}
		\begin{split}
			\sup\limits_{k} \|\Delta x_k\| \leq \bigg\| \bigg[ \sup\limits_{k} \|\Delta x_k^{(i)}\| \bigg]_{i=1}^M \bigg\| \leq	\sqrt{M} \sup\limits_{k} \|\Delta x_k\|.
		\end{split}
		\label{eq:gen_rel}
	\end{equation}
	By using the first and second bound in (\ref{eq:gen_rel}), defining ${h := \bar{g} b \max\nolimits_{i} p_i}$, and considering the time $t = \xi N + l_0$, we can derive
	\begin{align}
		&\|\Delta x_{t}\| \leq \sup\limits_{k \in \mathbb{I}_{\geq \xi N}} \| \Delta x_k \| \leq \bigg\| \bigg[ \sup\limits_{k \in \mathbb{I}_{\geq \xi N}} \| \Delta x_{k}^{(i)} \| \bigg]_{i=1}^M \bigg\| \nonumber \\
		& \stackrel{(\ref{eq:main_bound_norm})}{\leq} h \sigma^t \| \Delta x_0 \| \nonumber \\
		& + \bar{b} \sqrt{M}  \bigg( \tilde{q}_i  \max\limits_{k \in \mathbb{I}_{[0,t-1]}} \| \Delta w_k \|  + \tilde{r}_i \max\limits_{k \in \mathbb{I}_{[0,t-1]}} \| \Delta y_k \| \bigg).
		\label{eq:final_bound}
	\end{align}
	Note that due to causality, the maximum of $\|\Delta w \|$ and $\|\Delta y \|$ is taken up to $t-1$. Moreover, (\ref{eq:final_bound}) corresponds to exponential i-IOSS of the overall system in classical asymptotic gain formulation using sum of max-norms. For asymptotic (i.e., not necessarily exponential) i-IOSS, it is shown in \mbox{\cite[Prop. 2.5]{allan_21}} that this is equivalent to i-IOSS in time-discounted sum formulation. This can be straightforwardly specialized to the exponential case, i.e., Definition \ref{def:IOSS}, which completes the proof.
\end{proof}

\newtheorem{remark}{Remark}
\begin{remark}
	The bound on the disturbance and output term in (\ref{eq:final_bound}) depends on the number of subsystems $M$. This dependency arises, because we impose the explicit structure in terms of the max-norm of the stacked disturbances $w$ and outputs $y$, instead of the norm of the stacked max-norms of the individual disturbances $w^{(i)}$ and outputs $y^{(i)}$. Note that a tighter bound on $\|\Delta x_t\|$ with gains independent of $M$ can be derived analogous to (\ref{eq:final_bound}), but by directly using the right-hand side of (\ref{eq:main_bound_norm}) and not the second inequality of (\ref{eq:gen_rel}).
\end{remark}

\subsection{Lyapunov characterization}
\label{sec:lyap}

Next, we show that under a small-gain condition on the couplings between the subsystems, an exponential \mbox{i-IOSS} Lyapunov function for the overall system (\ref{eq:sys}) can be constructed from such Lyapunov functions for the individual subsystems.

\begin{assumption}
	Each subsystem (\ref{eq:sub}) admits a local exponential i-IOSS Lyapunov function $V^{(i)}(x_t^{(i)},\tilde{x}_t^{(i)}) : {\mathbb{R}^{n_i} \times \mathbb{R}^{n_i} \rightarrow \mathbb{R}_{\geq 0}}$ with $P_1^{(i)}, \; P_2^{(i)} \succ 0$, $\lambda_i \in (0,1)$, $\gamma_{ij} > 0$ and $Q_i,R_i \succeq 0$ such that
	\begin{subequations}
		\begin{align}
			\|x_t^{(i)} - \tilde{x}_t^{(i)}\|_{P_1^{(i)}}^2 \leq V^{(i)}(x_t^{(i)},\tilde{x}_t^{(i)}) \leq \|x_t^{(i)} - \tilde{x}_t^{(i)}\|_{P_2^{(i)}}^2,
		\end{align}
		\begin{align}
			V^{(i)}(x_{t+1}^{(i)},\tilde{x}_{t+1}^{(i)}) & - V^{(i)}(x_t^{(i)},\tilde{x}_t^{(i)}) \leq -\lambda_i V^{(i)}(x_t^{(i)},\tilde{x}_t^{(i)}) \nonumber \\
			& + \| w_t^{(i)} - \tilde{w}_t^{(i)} \|_{Q_i}^2 + \| y_t^{(i)} - \tilde{y}_t^{(i)} \|_{R_i}^2 \nonumber \\
			& + \sum_{j \in \mathcal{N}_i}^{} \gamma_{ij} V^{(j)}(x^{(j)}, \tilde{x}^{(j)})
			\label{eq:lyap_ass}
		\end{align}
	\end{subequations}
	\label{ass:lyap}
	for all times $t \in \mathbb{I}_{\geq 0}$, all initial conditions $x_0^{(i)}, \; \tilde{x}_0^{(i)} \in \mathbb{R}^{n_i}$, all disturbances $w^{(i)}, \; \tilde{w}^{(i)} \in \mathbb{R}^{q_i}$, and all ${x^{(j)}, \; \tilde{x}^{(j)} \in \mathbb{R}^{n_j}}$, where $x_{t+1}^{(i)} = f_i(x_t^{(i)}, u_t^{(i)}, w_t^{(i)}, z_t^{(i)})$, $\tilde{x}_{t+1}^{(i)} = f_i(\tilde{x}_t^{(i)}, u_t^{(i)}, \tilde{w}_t^{(i)}, \tilde{z}_t^{(i)})$, $y_t^{(i)} = h_i(x_t^{(i)}, u_t^{(i)}, w_t^{(i)}, z_t^{(i)})$, $\tilde{y}_t^{(i)} = h_i(\tilde{y}_t^{(i)}, u_t^{(i)}, \tilde{w}_t^{(i)}, \tilde{z}_t^{(i)})$, $t \in \mathbb{I}_{\geq 0}$.
\end{assumption}

In order to show that the overall system admits an exponential i-IOSS Lyapunov function, we follow the approach in \cite{dashkovskiy_11}, where an ISS Lyapunov function for the overall system is determined by using a small-gain condition. This is adapted to i-IOSS by considering the incremental version and taking the outputs into account. Before stating the theorem, we introduce the matrices  $\Lambda := \mathrm{diag}(\lambda_i, \cdot \cdot \cdot, \lambda_M)$ and $\Gamma := (\gamma_{ij})_{i,j=1, \cdot \cdot \cdot, M}$, where $\lambda_i$ and $\gamma_{ij}$ are the coefficients from Assumption \ref{ass:lyap}.

\begin{theorem}[exponential i-IOSS Lyapunov function]
	Let Assumption \ref{ass:lyap} hold for each subsystem (\ref{eq:sub}). Then, the overall system (\ref{eq:sys}) admits an exponential i-IOSS Lyapunov function if the small-gain condition $\rho(\Lambda^{-1}\Gamma) < 1$ is satisfied.
	\label{theo:lyap}
\end{theorem}

\begin{proof}
	By stacking all inequalities (\ref{eq:lyap_ass}) from index $1$ to $M$, we arrive at the compact form
	\begin{align}
		\big[ V^{(i)}(x_{t+1}^{(i)},&\tilde{x}_{t+1}^{(i)}) - V^{(i)}(x_t^{(i)},\tilde{x}_t^{(i)}) \big]_{i=1}^M \leq \nonumber \\
		& (-\Lambda+\Gamma) \big[ V^{(i)}(x_t^{(i)},\tilde{x}_t^{(i)}) \big]_{i=1}^M \nonumber \\
		& + \big[ \|\Delta w_t^{(i)}\|_{Q_i}^2 \big]_{i=1}^M + \big[ \|\Delta y_t^{(i)}\|_{R_i}^2 \big]_{i=1}^M.
		\label{eq:lyap_comp}
	\end{align}
	Now, assuming that the small-gain condition $\rho(\Lambda^{-1}\Gamma) < 1$ is satisfied, according to \cite[Lemma 3.1]{dashkovskiy_11} there exists a strictly positive vector $\mu \in \mathbb{R}^M$ with $\mu_i > 0$ for all $i \in \mathbb{I}_{[1,M]}$ such that $\mu^T(-\Lambda+\Gamma)<0$ (component-wise). Next, we define
	\begin{equation}
		\begin{split}
			V_{\Sigma}(x_t,\tilde{x}_t) := \mu^T \big[ V^{(i)}(x_t^{(i)},\tilde{x}_t^{(i)}) \big]_{i=1}^M
		\end{split}
		\label{eq:lyap_func}
	\end{equation}
	as exponential i-IOSS Lyapunov function for the overall system. Note that $V_{\Sigma}(x_t,\tilde{x}_t)$ can be bounded by
	\begin{equation}
		\|x_t - \tilde{x}_t\|_{P_1}^2 \leq V_{\Sigma}(x_t,\tilde{x}_t) \leq \|x_t - \tilde{x}_t\|_{P_2}^2
	\end{equation}
	with $P_1 = \mathrm{diag} \big( \mu_1 P_1^{(1)}, \cdot \cdot \cdot, \mu_M P_1^{(M)} \big)$ and $P_2 = \mathrm{diag} \big( \mu_1 P_2^{(1)}, \cdot \cdot \cdot, \mu_M P_2^{(M)} \big)$. From (\ref{eq:lyap_func}) and (\ref{eq:lyap_comp}), we obtain
	\begin{align}
		V_{\Sigma}&(x_{t+1},\tilde{x}_{t+1}) - V_{\Sigma}(x_t,\tilde{x}_t) = \nonumber \\
		& \mu^T \big[ V^{(i)}(x_{t+1}^{(i)},\tilde{x}_{t+1}^{(i)}) -  V^{(i)}(x_t^{(i)}, \tilde{x}_t^{(i)}) \big]_{i=1}^M \leq \nonumber \\
		& \mu^T (-\Lambda+\Gamma) \big[ V^{(i)}(x_t^{(i)}, \tilde{x}_t^{(i)}) \big]_{i=1}^M \nonumber \\
		& + \mu^T \big[\|\Delta w_t^{(i)}\|_{Q_i}^2\big]_{i=1}^M + \mu^T \big[\|\Delta y_t^{(i)}\|_{R_i}^2\big]_{i=1}^M.
	\end{align}
	Next, we define $H := \mu^T(-\Lambda+\Gamma)$. Note that $H$ is a vector of dimension $M$ with strictly negative entries, i.e., $H_i < 0$ for all $i \in \mathbb{I}_{[1,M]}$.
	Defining further
	\begin{align}
		\lambda_\Sigma & := - \max\limits_{i \in \mathbb{I}_{[1,M]}} \frac{H_i}{\mu_i} \nonumber \\
		& = - \max\limits_{i \in \mathbb{I}_{[1,M]}} \bigg( \sum_{j \in \mathcal{N}_i}^{} \frac{\mu_j}{\mu_i} \gamma_{ij} - \lambda_i \bigg) \in (0,1),
	\end{align} we can write
	\begin{align}
		V_{\Sigma}(x_{t+1},\tilde{x}_{t+1}) & - V_{\Sigma}(x_t,\tilde{x}_t) \leq -\lambda_\Sigma V_{\Sigma}(x_t,\tilde{x}_t) \nonumber \\
		& + \|\Delta w_t\|_{Q_\Sigma} + \|\Delta y_t\|_{R_\Sigma}
		\label{eq:lyap_final}
	\end{align}
	with $Q_\Sigma = \mathrm{diag}(\mu_1 Q_1, \cdot \cdot \cdot, \mu_M Q_M) \succeq 0$ and $R_\Sigma = \mathrm{diag}(\mu_1 R_1, \cdot \cdot \cdot, \mu_M R_M) \succeq 0$. Inequality (\ref{eq:lyap_final}) exhibits the same structure as (\ref{eq:lyap_def}), which shows that $V_{\Sigma}(x,\tilde{x})$ is an exponential i-IOSS Lyapunov function for the overall system (\ref{eq:sys}) according to Definition \ref{def:lyap}, which completes the proof.
\end{proof}

\subsection{Verification}
\label{sec:verif}

In this section, we derive LMI conditions to verify satisfaction of Assumptions \ref{ass:IOSS} and \ref{ass:lyap}. In doing so, we extend the centralized verification in \cite{schiller_23} to a distributed one by considering the couplings between the subsystems. For brevity, we omit the index $i$ in this section, since we do not consider the overall system, but only the single subsystems for all $i \in \mathbb{I}_{[1,M]}$. Moreover, we use $\nu, \; \omega, \; \psi$ instead of $u^{(i)}, \; w^{(i)}, \; y^{(i)}$, respectively, to distinguish subsystem variables from those of the overall system. However, for the subsystem state $x^{(i)}$ we retain this notation to separate it from the coupling states $x^{(j)}$.

\newtheorem{proposition}{Proposition}
\begin{proposition}
	Consider the subsystems (\ref{eq:sub}) with $h_i \in \mathbb{R}^{p_i}$ being affine in $x^{(i)} \in \mathbb{R}^{n_i}$ for all $i \in \mathbb{I}_{[1,M]}$. Let
	\begin{align}
		& \begin{pmatrix}
			A^T \tilde{P} A - \tilde{\eta} \tilde{P} & A^T \tilde{P} B & A^T \tilde{P} E \\
			B^T \tilde{P} A & B^T \tilde{P} B - \tilde{Q} & B^T \tilde{P} E \\
			E^T \tilde{P} A & E^T \tilde{P} B & E^T \tilde{P} E - \tilde{G}
		\end{pmatrix}
		\nonumber \\
		& -
		\begin{pmatrix}
			C^T \tilde{R} C & C^T \tilde{R} D & C^T \tilde{R} F \\
			D^T \tilde{R} C & D^T \tilde{R} D & D^T \tilde{R} F \\
			F^T \tilde{R} C & F^T \tilde{R} D & F^T \tilde{R} F
		\end{pmatrix} \preceq 0
		\label{eq:lmi}
	\end{align}
	be satisfied for all $x^{(i)} \in \mathbb{R}^{n_i}$, $\nu \in \mathbb{R}^{m_i},$ $\omega \in \mathbb{R}^{q_i}$, and $z^{(i)} \in \mathbb{R}^{s_i}$ with $\tilde{\eta} \in (0,1), \, \tilde{P} \succ 0, \; \tilde{Q}, \tilde{R}, \tilde{G} \succeq 0$, where the matrices $A$, $B$, $C$, $D$, $E$, and $F$ are given by
	\begin{align}
		A &= \frac{\partial f_i}{\partial x^{(i)}}(x^{(i)},\nu,\omega,z^{(i)}), \; B = \frac{\partial f_i}{\partial \omega}(x^{(i)},\nu,\omega,z^{(i)}), \nonumber \\
		C &= \frac{\partial h_i}{\partial x^{(i)}}x^{(i)},\nu,\omega,z^{(i)}), \; D = \frac{\partial h_i}{\partial \omega}(x^{(i)},\nu,\omega,z^{(i)}), \nonumber \\
		E &= \frac{\partial f_i}{\partial z^{(i)}}(x^{(i)},\nu,\omega,z^{(i)}), \; F = \frac{\partial h_i}{\partial z}(x^{(i)},\nu,\omega,z^{(i)}).
		\label{eq:lmi_deriv}
	\end{align}
	Then, Assumptions \ref{ass:IOSS} and \ref{ass:lyap} hold.
	\label{prop:verif}
\end{proposition}

\begin{proof}
	By a straightforward extension of the proof of \cite[Corollary 3]{schiller_23}, it can be shown that satisfaction of (\ref{eq:lmi}) implies that the inequality
	\begin{align}
		\|x_{t+1}^{(i)}-\tilde{x}_{t+1}^{(i)}\|_{\tilde{P}}^2 & \leq \tilde{\eta} \|x_t^{(i)}-\tilde{x}_t^{(i)}\|_{\tilde{P}}^2 + \|\omega_t-\tilde{\omega}_t\|_{\tilde{Q}}^2 \nonumber \\
		& + \|\psi_t-\tilde{\psi}_t\|_{\tilde{R}}^2 + \|z_t^{(i)}-\tilde{z}_t^{(i)}\|_{\tilde{G}}^2
		\label{eq:verif}
	\end{align}
	holds for all times $t \in \mathbb{I}_{\geq 0}$, all initial conditions $x_0^{(i)}, \; \tilde{x}_0^{(i)} \in \mathbb{R}^{n_i}$, all $z_t^{(i)}, \; \tilde{z}_t^{(i)} \in \mathbb{R}^{s_i}$, and all disturbances $\omega_t, \;\tilde{\omega}_t \in \mathbb{R}^{q_i}$, where $x_{t+1}^{(i)} = f_i(x_t^{(i)}, \nu_t, \omega_t, z_t^{(i)})$, $\tilde{x}_{t+1}^{(i)} = f_i(\tilde{x}_t^{(i)}, \nu_t, \tilde{\omega}_t, \tilde{z}_t^{(i)})$, $\psi_t = h_{i}(x_t^{(i)}, \nu_t, \omega_t,z_t^{(i)})$, $\tilde{\psi}_t = h_{i}(\tilde{x}_t^{(i)}, \nu_t, \tilde{\omega}_t, \tilde{z}_t^{(i)})$, $t \in \mathbb{I}_{\geq 0}$. \\
	In order to show that Assumption \ref{ass:lyap} is satisfied, we choose $\gamma_{ij} > 0$ for all $j \in \mathcal{N}_i$ such that
		\begin{equation}
			\|z^{(i)}-\tilde{z}^{(i)}\|_{\tilde{G}}^2 \leq \sum_{j \in \mathcal{N}_i}^{} \gamma_{ij} \|x^{(j)}-\tilde{x}^{(j)}\|_{\tilde{P}}^2
			\label{eq:z_to_x_2}
		\end{equation}
		holds for all $z^{(i)}, \tilde{z}^{(i)} \in \mathbb{R}^{s_i}$ and all $x^{(j)}, \tilde{x}^{(j)} \in \mathbb{R}^{n_j}$. Note that the existence of such $\gamma_{ij}$ is given due to the fact that $G \succeq 0$ and $P \succ 0$.
		Moreover, we set $\lambda = 1-\tilde{\eta}$ and use $V(x^{(i)},\tilde{x}^{(i)}) = \|x^{(i)} - \tilde{x}^{(i)}\|_{\tilde{P}}^2$ as exponential i-IOSS Lyapunov function to obtain the same structure of (\ref{eq:lyap_ass}). \\
		For Assumption \ref{ass:IOSS}, we choose similar to (\ref{eq:z_to_x_2}) suitable $\tilde{g}_{ij} > 0$ for all $j \in \mathcal{N}_i$ such that
		\begin{equation}
			\|z^{(i)}-\tilde{z}^{(i)}\|_{\tilde{G}}^2 \leq \sum_{j \in \mathcal{N}_i}^{} \tilde{g}_{ij} \|x^{(j)}-\tilde{x}^{(j)}\|^2
			\label{eq:z_to_x_1}
		\end{equation}
		is satisfied for all $z^{(i)}, \tilde{z}^{(i)} \in \mathbb{R}^{s_i}$ and all $x^{(j)}, \tilde{x}^{(j)} \in \mathbb{R}^{n_j}$.
		Then, inserting (\ref{eq:z_to_x_1}) in (\ref{eq:verif}) and applying (\ref{eq:verif}) repeatedly, we can write
		\begin{align}
			\|x_{t}^{(i)}-\tilde{x}_{t}^{(i)}\| & \leq \frac{1}{\sqrt{\lambda_{\min}(\tilde{P})}} \bigg( \eta^t \sqrt{\lambda_{\max}(\tilde{P})} \|x_0^{(i)}-\tilde{x}_0^{(i)}\| \nonumber \\
			& + \sum_{k=0}^{t-1} \eta^{t-k-1} \bigg( \sqrt{\lambda_{\max}(\tilde{Q})} \|\omega_k-\tilde{\omega}_k\| \nonumber \\
			& + \sqrt{\lambda_{\max}(\tilde{R})} \|\psi_k-\tilde{\psi}_k\| \nonumber \\
			& + \sum_{j \in \mathcal{N}_i}^{}\sqrt{\tilde{g}_{ij}} \|x_k^{(j)}-\tilde{x}_k^{(j)}\| \bigg) \bigg)
			\label{eq:verif_final}
		\end{align}
		with $\eta = \sqrt{\tilde{\eta}}$, which corresponds to (\ref{eq:IOSS}). The coupling gains in (\ref{eq:IOSS}) can then be computed by
		\begin{equation}
			g_{ij} = \sqrt{\frac{\tilde{g}_{ij}}{\lambda_{\min}(\tilde{P})}}.
			\label{eq:g_ij_final}
		\end{equation}
\end{proof}

Some remarks are in order.

\begin{remark}
	\label{rem:numeric}
	For a fixed $\tilde{\eta} \in (0,1)$, (\ref{eq:lmi}) results in an infinite set of LMIs with decision variables $\tilde{P} \succ 0, \; \tilde{Q}, \tilde{R}, \tilde{G} \succeq 0$, which can be efficiently verified using sum-of-squares tools or linear parameter-varying embeddings. In practice, we usually consider compact sets on the system states, inputs and disturbances, resulting in a finite set of LMIs that can be verified via gridding and semidefinite programming (SDP). Similarly, we can determine the optimal, i.e., smallest, variables $\gamma_{ij}$ in (\ref{eq:z_to_x_2}) and $\tilde{g}_{ij}$ in (\ref{eq:z_to_x_1}) via SDP.
\end{remark}

\begin{remark}
	A conservative choice of the gains in (\ref{eq:z_to_x_2}) and (\ref{eq:z_to_x_1}) is $\gamma_{ij} = \lambda_{\max}(\tilde{G}) \lambda_{\min}(\tilde{P})^{-1}$ and $\tilde{g}_{ij} = \lambda_{\max}(\tilde{G})$, respectively, for all $j \in \mathcal{N}_i$, $i \in \mathbb{I}_{[1,M]}$,
	which leads to the relation $g_{ij} = \sqrt{\gamma_{ij}}$ in (\ref{eq:g_ij_final}).
	Hence, we conclude that $\gamma_{ij} < g_{ij}$ if $\lambda_{\max}(\tilde{G}) \lambda_{\min}(\tilde{P})^{-1} < 1$. In order to satisfy the small-gain condition, small values of $\gamma_{ij}$ and $g_{ij}$ are desirable. Therefore, the effort during the verification is to keep the expression $\lambda_{\max}(\tilde{G}) \lambda_{\min}(\tilde{P})^{-1}$ as small as possible. Moreover, the gains $\gamma_{ij}$ and $\tilde{g}_{ij}$ are scaled with $(1-\tilde{\eta})^{-1}$ and ${(1-\sqrt{\tilde{\eta}})^{-1}}$, respectively. Since $\tilde{\eta} \in (0,1)$, the gains of the Lyapunov-based formulation are subject to a smaller scaling. Overall, these insights indicate a less conservative small-gain condition in Lyapunov coordinates, which is also observed in the numerical example in Section~\ref{sec:example}.
\end{remark}

\begin{remark}
	Direct verification of system-theoretic properties using LMIs suffers from the curse of dimensionality and rapidly becomes intractable. We propose a distributed verification approach in which both the number and the dimension of the LMIs to be solved depend only on the individual subsystems. If the subsystems are drawn from a finite set of different dynamics, this framework enables the verification of detectability for infinite-dimensional networks as $M \rightarrow \infty$ by solving only a finite number of LMIs, which will be demonstrated in Section~\ref{sec:example}.
\end{remark}

\section{EXAMPLE}
\label{sec:example}

To illustrate our results, the verification of a distributed system is carried out in this chapter.
In particular, a train with $M$ carriages is modeled as a mass-spring-damper system (cf. \cite{uyulan_20}). Here, we consider nonlinear damping with a cubic term, where each carriage is a subsystem, linked to the forward and rear carriage. Therefore, every subsystem has two couplings except for subsystems $i=1$ and $i=M$, i.e., the first and last carriage of the train, which exhibit one coupling. The dynamics for the position $x_1^{(i)}$ of each carriage is described by

\begin{equation}
	\begin{split}
		x_{1,t+1}^{(i)} = x_{1,t}^{(i)} + \delta x_{2,t}^{(i)} + w_{1,t}^{(i)} \quad \forall i \in \mathbb{I}_{[1,M]},
	\end{split}
\end{equation}

while the velocity $x_2^{(i)}$ is modeled as

\begin{subequations}
	\begin{align}
		x_{2,t+1}^{(1)} & = x_{2,t}^{(1)} + \frac{\delta}{m} \big(F_t +   k(x_{1,t}^{(2)}-x_{1,t}^{(1)}) \nonumber \\
		& + d(x_{2,t}^{(2)}-x_{2,t}^{(1)})^3 \big) + w_{2,t}^{(1)},
	\end{align}
	\begin{align}
		x_{2,t+1}^{(i)} & = x_{2,t}^{(i)} + \frac{\delta}{m} \big(  k(x_{1,t}^{(i+1)}-x_{1,t}^{(i)}) \nonumber \\
		& + d(x_{2,t}^{(i+1)}-x_{2,t}^{(i)})^3 +  k(x_{1,t}^{(i-1)}-x_{1,t}^{(i)}) \nonumber \\
		& + d(x_{2,t}^{(i-1)}-x_{2,t}^{(i)})^3 \big) + w_{2,t}^{(i)} \quad \forall i \in \mathbb{I}_{[2,M-1]},
	\end{align}
	\begin{align}
		x_{2,t+1}^{(M)} & = x_{2,t}^{(M)} + \frac{\delta}{m} \big(k(x_{1,t}^{(M-1)}-x_{1,t}^{(M)}) \nonumber \\
		& + d(x_{2,t}^{(M-1)}-x_{2,t}^{(M)})^3 \big)+ w_{2,t}^{(M)},
	\end{align}
\end{subequations}

with input $F_t \in \mathbb{R}$ and positive constants $\delta, \; m, \; k, \; d > 0$. The output equation is given by

\begin{equation}
	y_t^{(i)} = x_{1,t}^{(i)} + w_{3,t}^{(i)} \quad \forall i \in \mathbb{I}_{[1,M]},
\end{equation}

i.e., only the position of each carriage is measured. Here, we consider additive process and measurement noise ${w_t^{(i)} \in \mathbb{R}^3}$. Since the system dynamics of all carriages in the middle, i.e., $\forall i \in \mathbb{I}_{[2,M-1]}$, are equal, the verification of (\ref{eq:verif}) has to be performed only once for these subsystems. As for the first and the last subsystem, the same LMI condition (\ref{eq:lmi}) is derived, since the input $F_t$ vanishes when taking the derivatives (\ref{eq:lmi_deriv}). Therefore, only two verifications have to be performed for the overall system. Moreover, for the given example, the matrices $A$ and $E$ from (\ref{eq:lmi}) depend solely on the velocities, since the positions are linear in the system dynamics. As discussed in Remark \ref{rem:numeric}, the LMIs (\ref{eq:lmi}) need to be satisfied only on the considered compact set, i.e., for all physically attainable velocities of the train. Deriving the coupling gains $g_{ij}$ of (\ref{eq:IOSS}) and $\gamma_{ij}$ of (\ref{eq:lyap_ass}) by following the procedure in Section \ref{sec:verif} and setting up the gain matrices $G$ from (\ref{eq:stacked_sup})  and $A^{-1} \Gamma$ from (\ref{eq:lyap_comp}) as described in Sections \ref{sec:IOSS} and \ref{sec:lyap}, respectively, we can evaluate the small-gain condition, which is illustrated in Table \ref{tab:sgc}.
As indicated in Section \ref{sec:verif}, the small-gain condition in Lyapunov coordinates is less conservative, since for $M=4$ the condition for the trajectory-based exponential i-IOSS formulation is already violated, i.e., $\rho(G) \geq 1$. Moreover, in the given example, an arbitrary number of carriages can be appended, while the small-gain condition remains satisfied (in Lyapunov coordinates), as it is bounded by the maximum row sum of the gain matrices. Therefore, we can conclude that the overall system is exponential i-IOSS independent of the number of subsystems $M$.

\begin{table}
	\centering
	\vspace{0.5em}
	\caption{Comparison of the small-gain condition depending on the number of subsystems}
	\begin{tabular}{|c|c|c|}
		\hline
		$\mathbf{M}$ &\textbf{Expon. i-IOSS (\ref{sec:IOSS})} & \textbf{Lyapunov char. (\ref{sec:lyap})} \\
		\hline
		3 & \cmark  &  \cmark \\
		\hline
		4 & \xmark  & \cmark  \\
		\hline
		$\infty$ & \xmark  & \cmark  \\
		\hline
	\end{tabular}
	\label{tab:sgc}
\end{table}

\section{CONCLUSION}
\label{sec:conclusion}

In this work, we have analyzed the exponential i-IOSS property of large-scale nonlinear systems. In particular, the distributed nature of these systems was exploited by imposing local assumptions to conclude together with a small-gain condition detectability of the overall system. In doing so, a global gain matrix was constructed, which describes the mutual dependencies between the subsystems and to which the small-gain condition was applied.
Furthermore, a distributed analysis in Lyapunov coordinates is provided, resulting in a quantitatively different small-gain condition that was found to be less conservative. Moreover, this work provides a distributed verification method for exponential i-IOSS of large-scale nonlinear systems. In particular, the size of the LMIs to be solved depends on the individual subsystem dimensions, and the number of LMIs depends on the number of different subsystem dynamics. In particular, as illustrated in the example, if only a finite number of different subsystem dynamics exist, this allows for a distributed verification of i-IOSS even if the number of subsystems is potentially infinite.



\bibliographystyle{IEEEtran}
\bibliography{reference}

\end{document}